\documentclass[sigplan,screen,10pt,nonacm]{acmart}

\usepackage{main}
\usepackage{mathpartir}
\usepackage{tikz}
\usepackage{tikzit}
\usepackage{accents}
\usepackage{catchfilebetweentags}
\usepackage{agda}
\usepackage{hyperref}
\usepackage{newunicodechar}
\usepackage{appendix}
\usetikzlibrary{bbox}

\usepackage[only,llbracket,rrbracket]{stmaryrd}

\usepackage[color=yellow,textwidth=.75in]{todonotes}
\usepackage{caption}
\usepackage{subcaption}
\usepackage[clock]{ifsym}
\usepackage{stackengine}
\usepackage[T1]{fontenc}

\makeatletter
\newcommand{\dotminus}{\mathbin{\text{\@dotminus}}}

\newcommand{\@dotminus}{%
  \ooalign{\hidewidth\raise1ex\hbox{.}\hidewidth\cr$\m@th-$\cr}%
}
\makeatother

\newunicodechar{λ}{\ensuremath{\mathnormal\lambda}}
\newunicodechar{∀}{\ensuremath{\mathnormal\forall}}
\newunicodechar{ℕ}{\ensuremath{\mathnormal{\Nat}}}
\newunicodechar{Π}{\ensuremath{\mathnormal{\Pi}}}
\newunicodechar{Σ}{\ensuremath{\mathnormal{\Sigma}}}
\newunicodechar{≡}{\ensuremath{\mathnormal{\equiv}}}
\newunicodechar{≤}{\ensuremath{\mathnormal{\le}}}
\newunicodechar{≥}{\ensuremath{\mathnormal{\ge}}}
\newunicodechar{∎}{\qed}
\newunicodechar{ʳ}{\ensuremath{\mathnormal{^{r}}}}
\newunicodechar{⁺}{\ensuremath{\mathnormal{^{+}}}}
\newunicodechar{ϕ}{\ensuremath{\mathnormal{\phi}}}
\newunicodechar{₀}{\ensuremath{\mathnormal{_{0}}}}
\newunicodechar{₁}{\ensuremath{\mathnormal{_{1}}}}
\newunicodechar{₂}{\ensuremath{\mathnormal{_{2}}}}
\newunicodechar{∧}{\ensuremath{\mathnormal{\wedge}}}
\newunicodechar{⇒}{\ensuremath{\mathnormal{\Rightarrow}}}
\newunicodechar{≢}{\ensuremath{\mathnormal{\not\equiv}}}
\newunicodechar{≟}{\ensuremath{\mathnormal{\stackrel{?}{=}}}}
\newunicodechar{∸}{\ensuremath{\mathnormal{\dotminus}}}
\newunicodechar{ℂ}{\ensuremath{\mathnormal{\mathcal{C}}}}
\newunicodechar{⟪}{{\guillemotleft}}
\newunicodechar{⟫}{{\guillemotright}}
\newunicodechar{≈}{\ensuremath{\mathnormal{\approx}}}
\newunicodechar{⊔}{\ensuremath{\mathnormal{\sqcup}}}
\newunicodechar{⊓}{\ensuremath{\mathnormal{\sqcap}}}
\newunicodechar{∷}{\ensuremath{\mathnormal{\mathord{::}}}}

\NewDocumentCommand\Calf{}{\textbf{calf}}

\title{A Verified Cost Analysis of Joinable Red-Black Trees}

\author{Runming Li}
\orcid{0000-0001-7600-9069}
\email{runmingl@andrew.cmu.edu}
\affiliation{
  \institution{Carnegie Mellon University}
  \streetaddress{5000 Forbes Ave.}
  \city{Pittsburgh}
  \state{PA}
  \postcode{15213}
  \country{USA}
}

\author{Harrison Grodin}
\orcid{0000-0002-0947-3520}
\email{hgrodin@cs.cmu.edu}
\affiliation{
  \institution{Carnegie Mellon University}
  \streetaddress{5000 Forbes Ave.}
  \city{Pittsburgh}
  \state{PA}
  \postcode{15213}
  \country{USA}
}

\author{Robert Harper}
\orcid{0000-0002-9400-2941}
\email{rwh@cs.cmu.edu}
\affiliation{
  \institution{Carnegie Mellon University}
  \streetaddress{5000 Forbes Ave.}
  \city{Pittsburgh}
  \state{PA}
  \postcode{15213}
  \country{USA}
}

\begin{document}

\begin{abstract}
  Ordered sequences of data, specified with a \emph{join} operation to combine sequences, serve as a foundation for the implementation of parallel functional algorithms.
  This abstract data type can be elegantly and efficiently implemented using balanced binary trees, where a join operation is provided to combine two trees and rebalance as necessary.
  In this work, we present a verified implementation and cost analysis of joinable red-black trees in \Calf{}, a dependent type theory for cost analysis.
  We implement red-black trees and auxiliary intermediate data structures in such a way that all correctness invariants are intrinsically maintained.
  Then, we describe and verify precise cost bounds on the operations, making use of the red-black tree invariants.
  Finally, we implement standard algorithms on sequences using the simple join-based signature and bound their cost in the case that red-black trees are used as the underlying implementation.
  All proofs are formally mechanized using the embedding of \Calf{} in the Agda theorem prover.
\end{abstract}

\maketitle

\section{Introduction}\label{sec:intro}

Ordered sequences of data are essential to the efficient implementation of parallel functional algorithms \cite{ab-algorithms}.
One common presentation of the signature for ordered sequences containing elements of type $\alpha$ is given in \cref{fig:seq-sig}.
This signature provides an abstract type $\seqty$ along with three operations:
\begin{enumerate}
  \item A constructor, $\seqempty$, that represents the empty sequence containing no data of type $\alpha$.
  \item A constructor, $\seqjoin$, that appends two sequences with an element of type $\alpha$ in between.
  \item
    A destructor, $\seqrec[\rho]$, that recurs over a sequence to produce an element of type $\rho$.
    An $\seqempty$ sequence is mapped to the argument of type $\rho$;
    a sequence $\seqjoin(s_1, a, s_2)$ is destructed using the argument of type $$\seqty \to \rho \to \alpha \to \seqty \to \rho \to \rho,$$
    plugging $s_1$ and $s_2$ in for the sequence arguments, $a$ in for the $\alpha$ argument, and the recursive calls in for the $\rho$ arguments.
\end{enumerate}
These three operations give rise to implementations of all algorithms on ordered sequences of data; some examples are shown in \cref{code:seq-examples}.

\begin{figure}
  \iblock{
    \mrow{\mathbf{type}~\seqty}
    \mrow{\seqempty : \seqty}
    \mrow{\seqjoin : \seqty \to \alpha \to \seqty \to \seqty}
    \mrow{\seqrec[\rho] : \rho \to (\seqty \to \rho \to \alpha \to \seqty \to \rho \to \rho) \to \seqty \to \rho}
  }
  \caption{Signature for ordered sequences containing elements of type $\alpha$.}
  \label{fig:seq-sig}
\end{figure}

Many implementations of this signature are possible, using data structures such as lists and trees.
When trees are used, the data in the sequence is taken to be the in-order traversal of the tree.
For parallel efficiency, balanced trees are a sensible choice \cite{blelloch-greiner:1995}: if the recursor $\seqrec[\rho]$ performs both recursive calls in parallel, it is worthwhile to rebalance during a $\seqjoin$ in preparation for an efficient use of $\seqrec[\rho]$ later.
As studied by \citet{Blelloch2016,Blelloch2022} and \citet{Sun2019}, when sequences are implemented as balanced binary trees, implementations of common auxiliary functions on sequences have efficient sequential and parallel cost.
For example, sequences may be used as an implementation of finite sets when the stored data is sorted.
Then, using $\seqempty$, $\seqjoin$, and $\seqrec$, bulk set operations such as union and intersection can be implemented with polylogarithmic span.

\begin{figure}
  \iblock{
    \mrow{\codename{Sum} : \seqty[\nat] \to \nat}
    \mrow{\codename{Sum} = \seqrec[\nat]~0~(\lambda~\textunderscore~n_1~~n~\textunderscore~n_2.\ n_1 + n + n_2)}
  }

  \iblock{
    \mrow{\codename{Map} : (\alpha \to \beta) \to \seqty[\alpha] \to \seqty[\beta]}
    \mrow{\codename{Map}~f = \seqrec[{\seqty[\beta]}]~\seqempty~(\lambda~\textunderscore~s_1~a~\textunderscore~s_2.\ \seqjoin~s_1~(f~a)~s_2)}
  }

  \iblock{
    \mrow{\codename{Reverse} : \seqty \to \seqty}
    \mrow{\codename{Reverse} = \seqrec[\seqty]~\seqempty~(\lambda~\textunderscore~s_1~a~\textunderscore~s_2.\ \seqjoin~s_2~a~s_1)}
  }

  \caption{Sample implementations of auxiliary functions on sequences, in terms of $\seqempty$, $\seqjoin$, and $\seqrec[\rho]$.}
  \label{code:seq-examples}
\end{figure}

\subsection{Red-black trees}\label{sec:intro-rbt}

Here, we consider the \emph{red-black tree} (RBT) data structure \cite{guibas-sedgewick:1978,okasaki:1999}, a flavor of binary search tree with an elegant functional description and cost analysis.
For our purposes, a binary tree is inductive data structure where each inhabitant is either a \emph{leaf} node carrying no data or a \emph{node} carrying a key and two other binary tree children.
A red-black tree is a binary tree satisfying the following invariants:
\begin{enumerate}\label{sec:invariants}
  \item every node is colored either $\colorred$ or $\colorblack$;
  \item every leaf is considered $\colorblack$;
  \item both children of a $\colorred$-colored node must be colored $\colorblack$;
  \item the number of $\colorblack$ nodes on any path from the root to a leaf (excluding the leaf), called the \emph{black height} of the tree, is the same.
\end{enumerate}
Following \citet{Blelloch2016,Blelloch2022}, we do not require that the root of a red-black tree be colored black.
In \cref{fig:rbt-example}, we show a sample red-black tree with black height of 1.

\begin{figure}
  \begin{tikzpicture}[nodes={draw,circle}, level distance=2\baselineskip,
    level 1/.style={sibling distance=3cm},
    level 2/.style={sibling distance=1.5cm},
    level 3/.style={sibling distance=0.75cm}]
    \node[fill=RBTRed] {$x_2$}
      child { node[fill=RBTBlack] {$x_1$}
        child { node[fill=RBTRed] {$x_0$}
          child { node[rectangle, fill=RBTBlack] {} }
          child { node[rectangle, fill=RBTBlack] {} }
        }
        child { node[rectangle, fill=RBTBlack] {} }
      }
      child { node[fill=RBTBlack] {$x_4$}
        child { node[fill=RBTRed] {$x_3$}
          child { node[rectangle, fill=RBTBlack] {} }
          child { node[rectangle, fill=RBTBlack] {} }
        }
        child { node[fill=RBTRed] {$x_5$}
          child { node[rectangle, fill=RBTBlack] {} }
          child { node[rectangle, fill=RBTBlack] {} }
        }
      }
      ;
  \end{tikzpicture}
  \caption{Sample red-black tree with black height of 1. Leaves are depicted as black squares, and nodes are depicted as red or black circles annotated with a key.}
  \label{fig:rbt-example}
\end{figure}
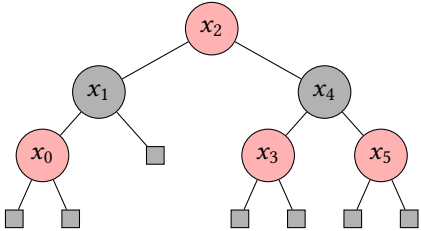

Traditionally, red-black trees have been used as binary search trees, storing data in sorted order.
Then, the primitive operations are insertion, lookup, and deletion, all of which have similar implementations.
However, as discussed by \citet{Blelloch2016,Blelloch2022}, this causes algorithms implemented using red-black trees to have poor parallel efficiency, since operations must be performed one-at-a-time.
Instead, \opcit{}, a $\seqjoin$ operation for red-black trees is given, combining two trees with a middle key and rebalancing as necessary to meet the red-black invariants and preserve the in-order traversal ordering.
In \cref{fig:rbts-smaller-join}, we show two sample red-black trees $t_1$ and $t_2$ which, when joined with $x_5$ in the middle, produce the tree $t$.

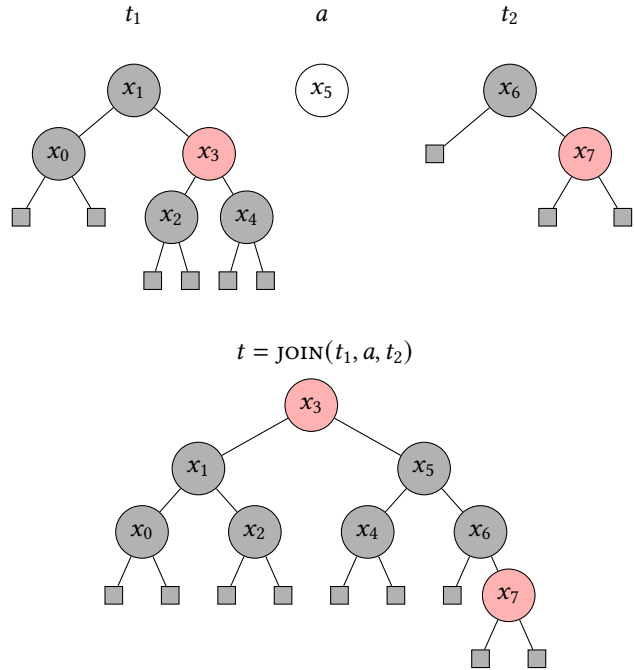
\begin{figure}
  \begin{tikzpicture}[nodes={draw,circle}, level distance=2\baselineskip,
    level 1/.style={sibling distance=2cm},
    level 2/.style={sibling distance=1cm},
    level 3/.style={sibling distance=0.5cm}]
    \node[draw=none] at (0, 1) {$t_1$};
    \node[draw=none] at (2.5, 1) {$a$};
    \node[draw=none] at (5, 1) {$t_2$};
    \node[fill=RBTBlack] at (0, 0) {$x_1$}
      child { node[fill=RBTBlack] {$x_0$}
        child { node[rectangle, fill=RBTBlack] {} }
        child { node[rectangle, fill=RBTBlack] {} }
      }
      child { node[fill=RBTRed] {$x_3$}
        child { node[fill=RBTBlack] {$x_2$}
          child { node[rectangle, fill=RBTBlack] {} }
          child { node[rectangle, fill=RBTBlack] {} }
        }
        child { node[fill=RBTBlack] {$x_4$}
          child { node[rectangle, fill=RBTBlack] {} }
          child { node[rectangle, fill=RBTBlack] {} }
        }
      }
      ;
    \node[fill=white!30] at (2.5, 0) {$x_5$}
      ;
    \node[fill=RBTBlack] at (5, 0) {$x_6$}
      child { node[rectangle, fill=RBTBlack] {} }
      child { node[fill=RBTRed] {$x_7$}
        child { node[rectangle, fill=RBTBlack] {} }
        child { node[rectangle, fill=RBTBlack] {} }
      }
      ;
  \end{tikzpicture}
  \vspace{1em}
  $$t = \seqjoin(t_1, a, t_2)$$
  \begin{tikzpicture}[nodes={draw,circle}, level distance=2\baselineskip,
    level 1/.style={sibling distance=3cm},
    level 2/.style={sibling distance=1.5cm},
    level 3/.style={sibling distance=0.75cm}]
    \node[fill=RBTRed]  {$x_3$}
      child { node[fill=RBTBlack] {$x_1$}
        child { node[fill=RBTBlack] {$x_0$}
          child { node[rectangle, fill=RBTBlack] {} }
          child { node[rectangle, fill=RBTBlack] {} }
        }
        child { node[fill=RBTBlack] {$x_2$}
          child { node[rectangle, fill=RBTBlack] {} }
          child { node[rectangle, fill=RBTBlack] {} }
        }
      }
      child { node[fill=RBTBlack] {$x_5$}
        child { node[fill=RBTBlack] {$x_4$}
          child { node[rectangle, fill=RBTBlack] {} }
          child { node[rectangle, fill=RBTBlack] {} }
        }
        child { node[fill=RBTBlack] {$x_6$}
          child { node[rectangle, fill=RBTBlack] {} }
          child { node[fill=RBTRed] {$x_7$}
            child { node[rectangle, fill=RBTBlack] {} }
            child { node[rectangle, fill=RBTBlack] {} }
        }
        }
      }
      ;
  \end{tikzpicture}
  \caption{Two red-black trees $t_1$ and $t_2$ along with the tree $t$ produced when they are joined with $a = x_5$ in the middle.}
  \label{fig:rbts-smaller-join}
\end{figure}

It is well-known that red-black trees intrinsically satisfying the above invariants can be defined inductively \cite{Licata2013,weirich:2014,wang-wang-chlipala:2017}:
\begin{enumerate}
  \item A black-colored RBT with black height $0$, a leaf, may always be formed.
  \item Let $t_1$ and $t_2$ are black-colored RBTs with black height $n$, and let $a$ be a key. Then, a red-colored RBT with black height $n$ may be formed.
  \item Let $t_1$ and $t_2$ be RBTs with black height $n$, and let $a$ be a key. Then, a black-colored RBT with black height $n + 1$ may be formed.
\end{enumerate}
We will use this presentation of red-black trees in our definitions and analysis.

\subsection{Mechanized cost analysis in \Calf{}}

The cost-aware logical framework (\Calf{})~\citep{Niu2022} is a dependent type theory for verifying the sequential and parallel cost and correctness of algorithms.
\Calf{} is based on the call-by-push-value paradigm~\citep{levy:2003}, separating computations (which may have an associated cost) from values.
Computation types are elements of the universe $\tpc$, whereas value types are elements of the universe $\tpv$.
Function types are computation types, where the input type is a value type and the output type is a computation type.
In this setting, the signature for ordered sequences from \cref{fig:seq-sig} is augmented to include $\U{-}$ and $\F{-}$ type constructors, explicitly moving between value and computation types; this change is rendered in \cref{fig:seq-sig-cbpv}.
\begin{figure}
  \iblock{
    \mrow{\seqty : \tpv}
    \mrow{\seqempty : \U{\seqty}}
    \mrow{\seqjoin : \U{\seqty \to \alpha \to \seqty \to \F{\seqty}}}
    \mrow{\seqrec[\rho] : \mathsf{U}(}
    \mrow{\quad \U{\rho}}
    \mrow{\quad \to \U{\seqty \to \U{\rho} \to \alpha \to \seqty \to \U{\rho} \to \rho}}
    \mrow{\quad \to \seqty \to \rho}
    \mrow{)}
  }
  \caption{Signature for ordered sequences containing elements of value type $\alpha : \tpv$, for computation types $\rho : \tpc$.}
  \label{fig:seq-sig-cbpv}
\end{figure}

In \Calf{}, the programmer includes cost annotations within algorithms, denoting an abstract notion of cost to later analyze.
In this work, we use the usual sequential-and-parallel cost model \citep[\S 6]{Niu2022}, where a cost is a pair of the sequential work and the parallel span as natural numbers.
To annotate a program with \costbox{$c$} (sequential and parallel) cost, we write $\step[c]$.

Originally, \citet{Niu2022} studied the implementation of sequential and parallel algorithms on concrete data structures in \Calf{}.
In subsequent work, \citet{grodin-harper:2023} consider the analysis of sequential-use data structures in this setting.
Here, we begin to investigate the implementation and analysis of parallel data structures in \Calf{}.

\subsection{Contribution}

In this work, we present an implementation of sequences using joinable red-black trees in \Calf{}.
The correctness of our implementation is intrinsically verified, and we perform a separate precise cost analysis in terms of the number of recursive calls.
Following \citet{Blelloch2016,Blelloch2022}, we implement a variety of sequence functions generically in the given primitives, and we analyze the cost of a simple function in the case the underlying implementation of the sequence type is the red-black tree data structure.

Our implementation and proofs are fully mechanized in Agda~\citep{norell:2009}, in which \Calf{} is embedded~\citep{Niu2022}.
We implement the mechanization of sequences and red-black trees in \texttt{Examples/Sequence.agda} and the corresponding \texttt{Examples/Sequence} directory.

\subsection{Related work}\label{sec:related}

Join-based balanced binary trees have been studied extensively by \citet{Blelloch2016,Blelloch2022}, and the joinable framework is unified by \citet{Sun2019}.

The correctness of red-black trees with their traditional sequential operations, such as single-element insertion, have been intrinsically (and extrinsically) verified in a variety of verification environments, including Agda \citep{Licata2013,weirich:2014}, Coq \citep{Appel2011,Appel:SF3}, and Isabelle \citep{nipkow:2023}.
However, these systems do not come equipped with a notion of cost, preventing the verification of the efficiency of these algorithms:
\begin{quote}
  Coq does not have a formal time--cost model for its execution, so we cannot verify [the] logarithmic running time [of insertion and lookup on red-black trees] in Coq.
  \attrib{\citep{Appel:SF3}}
\end{quote}

In another direction, the cost analysis of sequential operations on red-black trees has been verified in a resource-aware type theory \cite{wang-wang-chlipala:2017}.
However, this work does not verify the correctness of the data structure.

In this work, we verify both the correctness and cost of joinable red-black trees using an abstract cost model in \Calf{}; further explanation of and examples in the \Calf{} framework are presented in the original work of \citet{Niu2022}.
\section{Intrinsically-correct definitions}\label{sec:algo}

In this section, we describe a binary tree data type that structurally guarantees that the red-black invariants hold.
Then, we describe how it would be used to implement the sequence signature of \cref{fig:seq-sig-cbpv}; of particular interest is the implementation of the $\seqjoin$ algorithm.
Since our definitions will be well-typed, they will be intrinsically correct.
We work in \Calf{}, an extension of call-by-push-value, in which we distinguish value types in universe $\tpv$ from computation types in universe $\tpc$.

First, we define red-black trees as an indexed inductive type, as described in \cref{sec:intro-rbt}, guaranteeing that the red-black invariants are maintained; this definition of $\typename{irbt}_\alpha$ is given in \cref{code:irbt}.
We include an index storing the in-order traversal of the tree that we will to use to guarantee that well-typed definitions implement the desired behavior, specified in terms of lists.
Additionally, we define the type $\typename{rbt}_\alpha$ as the total space of the type family $\typename{irbt}_\alpha$, storing an arbitrary color, black-height, and in-order traversal along with an indexed red-black tree with those parameters.

\begin{figure}
  \iblock{
    \mrow{\kw{data}~\typename{irbt}_\alpha : \colorty \to \nat \to \listty{\alpha} \to \tpv~\kw{where}}
    \mrow{\quad \colorbox{RBTBlack}{$\conname{leaf}$} : \irbt{\colorblack}{\zero}{\nilex}}
    \mrow{\quad \colorbox{RBTRed}{$\conname{red}$} : (\irbt{\colorblack}{n}{l_1})~(a : \alpha)~(\irbt{\colorblack}{n}{l_2})}
    \mrow{\quad\quad \to \irbt{\colorred}{n}{(\listjoin{l_1}{a}{l_2})}}
    \mrow{\quad \colorbox{RBTBlack}{$\conname{black}$} : (\irbt{y_1}{n}{l_1})~(a : \alpha)~(\irbt{y_2}{n}{l_2})}
    \mrow{\quad\quad \to \irbt{\colorblack}{\suc{n}}{(\listjoin{l_1}{a}{l_2})}}
    \mrow{}
    \mrow{\rbt : \tpv}
    \mrow{\rbt = \sum_{y : \colorty} \sum_{n : \nat} \sum_{l : \listty{\alpha}} \irbt{y}{n}{l}}
  }
  \caption{Definition of indexed red-black trees as an indexed inductive type.}
  \label{code:irbt}
\end{figure}

Given these definitions, the goal is to implement the sequence signature of \cref{fig:seq-sig-cbpv}.
We choose ${\seqty = \rbt}$, define ${\seqempty = \ret{\ileaf}}$, and naturally implement $\seqrec$ via the induction principle for $\rbt$.
It remains, then, to define a computation $$\seqjoin : \rbt \to \alpha \to \rbt \to \F{\rbt},$$
which we consider in the remainder of this section.

\subsection{The \seqjoin{} algorithm}

The algorithm itself will follow \citet{Blelloch2016}, although we must ensure that the intrinsic structural properties are valid.%
\footnote{We omit a case listed by \citet{Blelloch2022} that our verification shows is impossible to reach.}
We recall its definition in \cref{alg:join}, adapting to our notation; it is defined in terms of auxiliary functions \codename{JoinRight} (and the symmetric \codename{JoinLeft}, which we henceforth elide), which we will consider in the next section.
Informally, the algorithm proceeds as follows:
\begin{enumerate}
  \item
    If both trees have equal height, simply construct a new node without rebalancing (\cref{fig:join-equal}).
    If possible, a red node is preferable.
  \item
    Otherwise, without loss of generality, assume $t_1$ has a larger black height than $t_2$.
    Then, use the \codename{JoinRight} auxiliary function to place $t_2$ on the right spine of $t_1$, rebalancing as necessary.
    The process may cause a single red-red violation at the root of the result tree.
    In that case, recolor the root to black (\cref{fig:join-violation}); otherwise, return the valid tree.
\end{enumerate}
This algorithm performs no recursive calls aside from those within \codename{JoinRight}, so no cost annotations are required by our cost model.
It remains, then, to define the type and implementation of \codename{JoinRight}.

\begin{algorithm}
  \begin{algorithmic}
  \Require
    \Statex $t_1 : \irbt{y_1}{n_1}{l_1}$
    \Statex $a : \alpha$
    \Statex $t_2 : \irbt{y_2}{n_2}{l_2}$
  \Ensure
    \Statex $\seqjoin(t_1, a, t_2) : \F{\sum_{y : \colorty} \sum_{n : \nat} \irbt{y}{n}{(\listjoin{l_1}{a}{l_2})}}$
  \Statex
    \Switch{$\Call{Compare}{n_1, n_2}$}
      \Case{$n_1 > n_2$}
        \State $t' \gets \Call{JoinRight}{t_1, a, t_2}$
        \Switch{$t'$}
          \Case{meets the invariants}
            \State \Return $t'$
          \EndCase
          \Case{has a red-red violation on the right}
            \State $\ired{t'_1}{a'}{t'_2} \gets t'$
            \State \Return $\iblack{t'_1}{a'}{t'_2}$
          \EndCase
      \EndCase
      \Case{$n_1 < n_2$}
        \State{$\cdots$} \Comment{symmetric, in terms of $\codename{JoinLeft}$}
      \EndCase
      \Case{$n_1 = n_2$}
        \If{$y_1 = \colorblack$ and $y_2 = \colorblack$}
          \State \Return $\ired{t_1}{a}{t_2}$
        \Else
          \State \Return $\iblack{t_1}{a}{t_2}$
        \EndIf
      \EndCase
  \end{algorithmic}
  \caption{%
    \codename{Join} algorithm for red-black trees \cite{Blelloch2022}.
    Since the color and black-height outputs may be inferred, we leave them implicit for readability.
  }
  \label{alg:join}
\end{algorithm}

\begin{figure}
  \begin{tikzpicture}[nodes={draw,circle}, level distance=2em,
    level 1/.style={sibling distance=2cm},
    level 2/.style={sibling distance=1cm},
    level 3/.style={sibling distance=0.5cm}]
    \node[draw=none] at (0, 1) {$t_1$};
    \node[draw=none] at (2.5, 1) {$a$};
    \node[draw=none] at (5, 1) {$t_2$};
    \node[fill=RBTRed] at (0, 0) {$x_2$}
      child { node[fill=RBTBlack] {$x_1$}
        child { node[fill=RBTRed] {$x_0$}
          child { node[rectangle, fill=RBTBlack] {} }
          child { node[rectangle, fill=RBTBlack] {} }
        }
        child { node[rectangle, fill=RBTBlack] {} }
      }
      child { node[fill=RBTBlack] {$x_4$}
        child { node[fill=RBTRed] {$x_3$}
          child { node[rectangle, fill=RBTBlack] {} }
          child { node[rectangle, fill=RBTBlack] {} }
        }
        child { node[fill=RBTRed] {$x_5$}
          child { node[rectangle, fill=RBTBlack] {} }
          child { node[rectangle, fill=RBTBlack] {} }
        }
      }
      ;
      \node at (2.5, 0) {$x_6$};
      \node[fill=RBTBlack] at (5, 0) {$x_8$}
      child { node[fill=RBTRed] {$x_7$}
        child { node[rectangle, fill=RBTBlack] {} }
        child { node[rectangle, fill=RBTBlack] {} }
      }
      child { node[fill=RBTRed] {$x_9$}
        child { node[rectangle, fill=RBTBlack] {} }
        child { node[rectangle, fill=RBTBlack] {} }
      }
      ;
  \end{tikzpicture}
  \vspace{1em}
  $$t = \seqjoin(t_1, a, t_2)$$
  \begin{tikzpicture}[nodes={draw,circle}, level distance=2em,
    level 1/.style={sibling distance=4cm},
    level 2/.style={sibling distance=2cm},
    level 3/.style={sibling distance=1cm},
    level 4/.style={sibling distance=0.5cm},
    ]
    \node[fill=RBTBlack] {$x_6$}
      child { node[fill=RBTRed] {$x_2$}
        child { node[fill=RBTBlack] {$x_1$}
          child { node[fill=RBTRed] {$x_0$}
            child { node[rectangle, fill=RBTBlack] {} }
            child { node[rectangle, fill=RBTBlack] {} }
          }
          child { node[rectangle, fill=RBTBlack] {} }
        }
        child { node[fill=RBTBlack] {$x_4$}
          child { node[fill=RBTRed] {$x_3$}
            child { node[rectangle, fill=RBTBlack] {} }
            child { node[rectangle, fill=RBTBlack] {} }
          }
          child { node[fill=RBTRed] {$x_5$}
            child { node[rectangle, fill=RBTBlack] {} }
            child { node[rectangle, fill=RBTBlack] {} }
          }
        }
      }
      child { node[fill=RBTBlack] {$x_8$}
        child { node[fill=RBTRed] {$x_7$}
          child { node[rectangle, fill=RBTBlack] {} }
          child { node[rectangle, fill=RBTBlack] {} }
        }
        child { node[fill=RBTRed] {$x_9$}
          child { node[rectangle, fill=RBTBlack] {} }
          child { node[rectangle, fill=RBTBlack] {} }
        }
      };
  \end{tikzpicture}
  \caption{Join of two trees with equal black heights.}
  \label{fig:join-equal}
\end{figure}
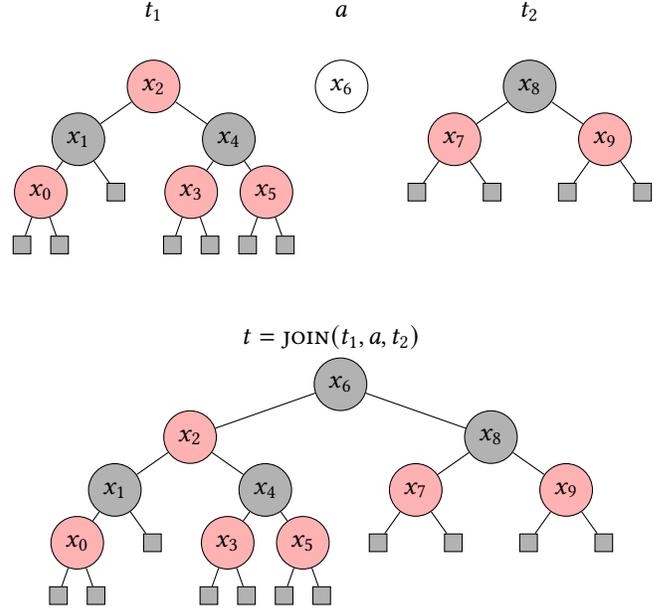

\begin{figure}
  \begin{tikzpicture}[nodes={draw,circle}, level distance=2em,
    level 1/.style={sibling distance=2cm},
    level 2/.style={sibling distance=1cm},
    level 3/.style={sibling distance=0.5cm}]
    \node[draw=none] at (0, 1) {$t_1$};
    \node[draw=none] at (2.5, 1) {$a$};
    \node[draw=none] at (5, 1) {$t_2$};
    \node[fill=RBTRed] at (0, 0) {$x_2$}
      child { node[fill=RBTBlack] {$x_1$}
        child { node[fill=RBTRed] {$x_0$}
          child { node[rectangle, fill=RBTBlack] {} }
          child { node[rectangle, fill=RBTBlack] {} }
        }
        child { node[rectangle, fill=RBTBlack] {} }
      }
      child { node[fill=RBTBlack] {$x_4$}
        child { node[fill=RBTRed] {$x_3$}
          child { node[rectangle, fill=RBTBlack] {} }
          child { node[rectangle, fill=RBTBlack] {} }
        }
        child { node[fill=RBTRed] {$x_5$}
          child { node[rectangle, fill=RBTBlack] {} }
          child { node[rectangle, fill=RBTBlack] {} }
        }
      }
      ;
      \node at (2.5, 0) {$x_6$};
      \node[fill=RBTRed] at (5, 0) {$x_7$}
        child { node[rectangle, fill=RBTBlack] {} }
        child { node[rectangle, fill=RBTBlack] {} }
      ;
  \end{tikzpicture}
  \vspace{1em}
  $$t = \codename{JoinRight}(t_1, a, t_2)$$
  \begin{tikzpicture}[nodes={draw,circle}, level distance=2em,
    level 1/.style={sibling distance=4cm},
    level 2/.style={sibling distance=2cm},
    level 3/.style={sibling distance=1cm},
    level 4/.style={sibling distance=0.5cm},
    ]
    \node[fill=RBTRed] at (0, 0) {$x_2$}
      child { node[fill=RBTBlack] {$x_1$}
        child { node[fill=RBTRed] {$x_0$}
          child { node[rectangle, fill=RBTBlack] {} }
          child { node[rectangle, fill=RBTBlack] {} }
        }
        child { node[rectangle, fill=RBTBlack] {} }
      }
      child[dashed] { node[fill=RBTRed,solid] {$x_6$}
        child[solid] { node[fill=RBTBlack] {$x_4$}
          child { node[fill=RBTRed] {$x_3$}
            child { node[rectangle, fill=RBTBlack] {} }
            child { node[rectangle, fill=RBTBlack] {} }
          }
          child { node[fill=RBTRed] {$x_5$}
            child { node[rectangle, fill=RBTBlack] {} }
            child { node[rectangle, fill=RBTBlack] {} }
          }
        }
        child[solid] { node[fill=RBTBlack] {$x_7$}
          child { node[rectangle, fill=RBTBlack] {} }
          child { node[rectangle, fill=RBTBlack] {} }
        }
      }
      ;
  \end{tikzpicture}
  \vspace{1em}
  $$t = \seqjoin(t_1, a, t_2)$$
  \begin{tikzpicture}[nodes={draw,circle}, level distance=2em,
    level 1/.style={sibling distance=4cm},
    level 2/.style={sibling distance=2cm},
    level 3/.style={sibling distance=1cm},
    level 4/.style={sibling distance=0.5cm},
    ]
    \node[fill=RBTBlack] at (0, 0) {$x_2$}
      child { node[fill=RBTBlack] {$x_1$}
        child { node[fill=RBTRed] {$x_0$}
          child { node[rectangle, fill=RBTBlack] {} }
          child { node[rectangle, fill=RBTBlack] {} }
        }
        child { node[rectangle, fill=RBTBlack] {} }
      }
      child { node[fill=RBTRed] {$x_6$}
        child { node[fill=RBTBlack] {$x_4$}
          child { node[fill=RBTRed] {$x_3$}
            child { node[rectangle, fill=RBTBlack] {} }
            child { node[rectangle, fill=RBTBlack] {} }
          }
          child { node[fill=RBTRed] {$x_5$}
            child { node[rectangle, fill=RBTBlack] {} }
            child { node[rectangle, fill=RBTBlack] {} }
          }
        }
        child { node[fill=RBTBlack] {$x_7$}
          child { node[rectangle, fill=RBTBlack] {} }
          child { node[rectangle, fill=RBTBlack] {} }
        }
      }
      ;
  \end{tikzpicture}
  \caption{Recoloring the root of a result tree from \codename{JoinRight} due to a red-red violation on the right, indicated by a dashed line.}
  \label{fig:join-violation}
\end{figure}
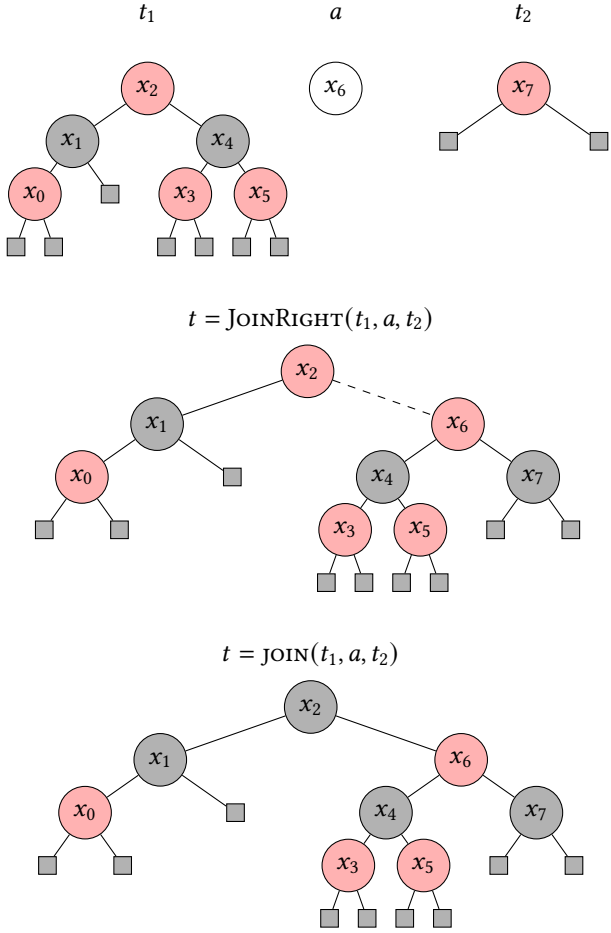

\subsection{The \codename{JoinRight} auxiliary algorithm}

As discussed previously, the \codename{JoinRight} algorithm has a relaxed specification: rather than guaranteeing a valid red-black tree, it allows a single red-red violation between the root of the result and its right child to propagate upwards.
We allow this violation only in the case that the first tree had a red root to begin with.

In order to represent this condition, we define an auxiliary data structure, an \emph{almost-right red-black tree}, abbreviated $\typename{arrbt}$, in \cref{code:arrbt}; our terminology is inspired by the ``almost tree'' of \citet{weirich:2014}.
A well-formed red-black tree always counts as an almost-right red-black tree; a $\colorred$-colored almost-right red-black tree may also be a violation, with a $\colorblack$-colored left child, key data, and another $\colorred$-colored right child.
Notably, a red-red violation for an almost-right red-black tree can only happen on the right spine, and only when the first tree originally had a red root.
We thereby define $\typename{arrbt}$ to be indexed by another color parameter called $\typename{leftColor}$, representing the color of the left tree from which it was created.
Therefore, when a violation happens, the $\typename{leftColor}$ must be $\colorred$.
Given this definition, we wish to define a computation
\iblock{
  \mrow{\codename{JoinRight} :}
  \mrow{\quad (\irbt{y_1}{n_1}{l_1})~(a : \alpha)~(\irbt{y_2}{n_2}{l_2}) \to}
  \mrow{\quad n_1 > n_2 \to}
  \mrow{\quad \F{\arrbt{y_1}{n_1}{(\listjoin{l_1}{a}{l_2})}}.}
}\noindent
Observe that the black height and left color of the result must match the first tree.
Also, notice that given such a definition of \codename{JoinRight}, the $\seqjoin$ implementation of \cref{alg:join} is well-typed and therefore correct.

\begin{figure}
  \iblock{
    \mrow{\kw{data}~\typename{arrbt}_\alpha : \colorty \to \nat \to \listty{\alpha} \to \tpv~\kw{where}}
    \mrow{\quad \conname{valid} : (\conname{leftColor} : \colorty)~(\irbt{y}{n}{l})}
    \mrow{\quad\quad \to \arrbt{\conname{leftColor}}{n}{l}}
    \mrow{\quad \conname{violation} : (\irbt{\colorblack}{n}{l_1})~(a : \alpha)~(\irbt{\colorred}{n}{l_2})}
    \mrow{\quad\quad \to \arrbt{\colorred}{n}{(\listjoin{l_1}{a}{l_2})}}
  }
  \caption{Definition of almost-right red-black trees, allowing for a red-red violation on the right when the color parameter (the color of the left tree from which it was created) is $\colorred$, as an indexed inductive type.}
  \label{code:arrbt}
\end{figure}

\begin{lemma}\label{lem:join}
  For all well-typed $t_1$, $a$, and $t_2$, it is the case that $\Call{Join}{t_1, a, t_2}$ is a valid red-black tree.
\end{lemma}

\begin{proof}
  We assume a well-typed implementation of \codename{JoinRight}, which is provided in \cref{alg:join-right} and proved correct in \cref{lem:left-black-red}.

  If $n_1 > n_2$, the call $\Call{JoinRight}{t_1, a, t_2}$ is made, returning an almost-right red-black tree.
  If this tree is valid, this tree is returned, as desired.
  Otherwise, if it has a red-red violation between the root and its right child, then the root is changed to black, causing all the red-black invariants to be satisfied.

  If $n_1 < n_2$, then a symmetric argument can be made.

  If $n_1 = n_2$, then the two trees may be joined by a red node if both are black or a black node otherwise.
  In either case, it forms a valid red-black tree.
\end{proof}

Now, it remains to give the \codename{JoinRight} algorithm to fulfill this specification.
Here, we diverge slightly from \citet{Blelloch2016,Blelloch2022} for ease of verification.
The algorithm presented \opcit{} allows for a triple-red violation on the right spine, albeit only in the base case.
Moreover, as noted by \citet[\S 3.2.2]{Sun2019}, the triple-red issue must be resolved one recursive call after the base case.
Therefore, we trade the more concise code and more complex specification for slightly more verbose code with a simpler specification.
We give our definition of \codename{JoinRight} in \cref{alg:join-right}.

\begin{algorithm}
  \begin{algorithmic}
  \Require
    \Statex $t_1 : \irbt{y_1}{n_1}{l_1}$
    \Statex $a : \alpha$
    \Statex $t_2 : \irbt{y_2}{n_2}{l_2}$
    \Statex $n_1 > n_2$
  \Ensure
    \Statex $\codename{JoinRight}(t_1, a, t_2) : \F{\arrbt{y_1}{n_1}{(\listjoin{l_1}{a}{l_2})}}$
  \Statex
    \Switch{$t_1$}
    \Case{$t_1 = \ired{t_{1,1}}{a_1}{t_{1,2}}$} \Comment{Case I.}
      \Step 1
      \State $\avalid{t'} \gets \Call{JoinRight}{t_{1,2}, a, t_2}$
      \Switch{$y'$, the color of $t'$}
      \Case{$y' = \colorred$}
        \State \Return $\aviolation{t_{1,1}}{a_1}{t'}$
      \EndCase
      \Case{$y' = \colorblack$}
        \State \Return $\avalid{\ired{t_{1,1}}{a_1}{t'}}$
      \EndCase{}
    \EndCase
    \Case{$t_1 = \iblack{t_{1,1}}{a_1}{t_{1,2}}$}
      \Switch{compare $n_1$ and $n_2$}
      \Case{$n_1 = n_2 + 1$}
        \Switch{$t_2$}
        \Case{$t_2 = \ired{t_{2,1}}{a_2}{t_{2,2}}$} \Comment{Case II.}
          \State \Return $\avalid{\ired{t_1}{a}{\iblack{t_{2,1}}{a_2}{t_{2,2}}}}$
        \EndCase
        \Case{$t_2 = \iblack{t_{2,1}}{a_2}{t_{2,2}}$}
          \Switch{$t_{1,2}$}
          \Case{$t_{1,2} = \ired{t_{1,2,1}}{a_{1,2}}{t_{1,2,2}}$} \Comment{Case III.}
            \State $x_1 \gets \iblack{t_{1,1}}{a_1}{t_{1,2,1}}$
            \State $x_2 \gets \iblack{t_{1,2,2}}{a}{t_2}$
            \State \Return $\avalid{\ired{x_1}{a_{1,2}}{x_2}}$
          \EndCase
          \Case{$t_{1,2} = \iblack{t_{1,2,1}}{a_{1,2}}{t_{1,2,2}}$} \Comment{Case IV.}
            \State $x_2 \gets \ired{t_{1,2}}{a}{t_2}$
            \State \Return $\avalid{\iblack{t_{1,1}}{a_1}{x_2}}$
          \EndCase
        \EndCase
      \EndCase
      \Case{$n_1 > n_2 + 1$}
        \Step 1
        \State $r \gets \Call{JoinRight}{t_{1,2}, a, t_2}$
        \Switch{$r$}
        \Case{$r = \avalid{t'}$}
          \State \Return $\avalid{\iblack{t_{1,1}}{a_1}{t'}}$ \Comment{Case V.}
        \EndCase
        \Case{$r = \aviolation{t'_1}{a'}{\ired{t'_{2,1}}{a'_2}{t'_{2,2}}}$}
          \State $x_1 \gets \iblack{t_{1,1}}{a_1}{t'_1}$
          \State $x_2 \gets \iblack{t'_{2,1}}{a'_2}{t'_{2,2}}$
          \State \Return $\avalid{\ired{x_1}{a'}{x_2}}$ \Comment{Case VI.}
        \EndCase
      \EndCase
    \EndCase
  \end{algorithmic}
  \caption{
    \codename{JoinRight} algorithm for red-black trees, based on \citet{Blelloch2022}.
    Cases are exhaustive, by the definitions of $\typename{irbt}_\alpha$ and $\typename{arrbt}_\alpha$, with the outer induction on $t_1$.
    Cost annotations are \costbox{highlighted}.
  }
  \label{alg:join-right}
\end{algorithm}

We claim that \codename{JoinRight} is a well-typed program with exhaustive casework, by the definitions of $\typename{irbt}_\alpha$ and $\typename{arrbt}_\alpha$.
Although our Agda mechanization verifies this fact, we include an informal proof below.

\begin{lemma} \label{lem:left-black-red}
  For all appropriate inputs $t_1$, $a$, and $t_2$, the call $\Call{JoinRight}{t_1, a, t_2}$ returns an almost-right red-black tree with black height $n_1$.
  In other words:
  \begin{enumerate}
    \item \label{lem:left-black} If $t_1$ is colored $\colorblack$, then $\Call{JoinRight}{t_1, a, t_2}$ is a valid red-black tree with the same black height as $t_1$.
    \item \label{lem:left-red} If $t_1$ is colored $\colorred$, then $\Call{JoinRight}{t_1, a, t_2}$ is an almost-right red-black tree (valid or with a red-red violation) with the same black height as $t_1$.
  \end{enumerate}
\end{lemma}
\begin{proof}
  We prove both items simultaneously by induction on $t_1$, following the structure of the code.
  \begin{enumerate}[label=\Roman*.]
    \item
      If $t_1$ is colored $\colorred$, we must prove \cref{lem:left-red}, and its children $t_{1,1}$ and $t_{1,2}$ must both be colored $\colorblack$.
      Moreover, $n_1 = n_{1,1} = n_{1,2} > n_2$.
      By induction, the result of the recursive call to $\Call{JoinRight}{t_{12}, a, t_2}$, $t'$, gives a valid red-black tree with black height $n_{1,2}$.
      We always return a $\colorred$ node whose left child is the $\colorblack$ subtree $t_{1,1}$ and whose right child is $t'$, which could be either red or black.
      Depending on the color of $t'$, we will either get a valid $\colorred$ tree or a red-red violation on the right spine, both of which are allowed as the result for \cref{lem:left-red}.

    \item
      If $t_1$ is colored $\colorblack$, we must prove \cref{lem:left-black}.
      If $n_1 = n_2 + 1$ and $t_2$ is colored $\colorred$, then $n_1 = n_{2,1} + 1 = n_{2,2} + 1$.
      Therefore, the returned tree is valid with black height $n_1$.

    \item
      This case is similar to the previous case, but $t_2$ is colored $\colorblack$.
      If $t_{1,2}$ is colored $\colorred$, then $n_{1,1} = n_{1,2} = n_{2,1} = n_{2,2} = n_2$.
      Therefore, the returned tree is valid with black height $n_1$.

    \item
      This case is similar to the previous case, but $t_{1,2}$ is colored $\colorblack$.
      Thus, $n_{1,1} = n_{1,2} = n_{1,2,1} = n_{1,2,2} = n_2$, so the returned tree is valid with black height $n_1$.

    \item
      If $t_1$ is colored $\colorblack$, we must prove \cref{lem:left-black}.
      Suppose $n_1 > n_2 + 1$.
      Then, $n_{1,1} + 1 = n_{1,2} + 1 = n_1 > n_2$.
      Regardless of the color of $t_{1,2}$, the inductive hypothesis applies.
      If the result $r$ is a valid red-black tree $t'$, then $t_{1,1}$ and $a_1$ can be combined at a $\colorblack$ node to create a valid red-black tree with black height $n_1$.

    \item
      This case is similar to the previous case, but the result $r$ indicates a red-red violation between the root and its right child.
      Then, a left-rotation is performed to give back a valid $\colorred$-colored red-black tree with black height $n_1$.
  \end{enumerate}
  In every case, the in-order traversal of the tree is clearly preserved, by inspection of the left-to-right order of the subtrees and keys.
\end{proof}

Thus, we have described the \seqjoin{} algorithm on red-black trees and intrinsically verified its correctness.
Based on the correctness of \codename{JoinRight}, we also get a straightforward bound on the black height of the tree produced by \seqjoin{}, matching the result of \citet{Blelloch2016,Blelloch2022}.

\begin{theorem}\label{thm:bheight-bound}
  Let $t_1$ and $t_2$ be red-black trees with black heights $n_1$ and $n_2$, respectively.
  Then, the black height of the red-black tree returned by $\seqjoin(t_1, a, t_2)$ is either $\max(n_1, n_2)$ or $1 + \max(n_1, n_2)$.
\end{theorem}

\Cref{thm:bheight-bound} does not affect the cost analysis of $\seqjoin$, but it does impact cost analysis for algorithms that use $\seqjoin$; therefore, it is also mechanized in the implementation.

For the purpose of correctness analysis, the cost annotations did not play a role.
In the next section, we will state and prove cost bounds on the \seqjoin{} and \codename{JoinRight} algorithms.
\section{Cost analysis}\label{sec:cost}

To analyze the cost of algorithms in \Calf{}, we attempt to bound the number of calls to $\step$.
In the subsequent development, we will count informally; in our mechanization, we use the definition $\isBounded{A}{e}{c}$ and associated lemmas from the \Calf{} standard library \citep{Niu2022}.
From this section onward, we annotate all mechanized results with their name as defined in the Agda implementation using the typewriter font, \eg{} \Mech{joinRight/is-bounded}.

\subsection{Cost of \codename{JoinRight}}

If a red-black tree has black height $n$, it has true height bounded by at most $2n + 1$: on top of every $\colorblack$ node, an additional $\colorred$ node may (optionally) be placed without affecting the black height.
Similar, then, to how an almost-right red-black tree weakens the invariants in the case of a $\colorred$ root, so too must the cost analysis weaken the cost bound given a $\colorred$ root.

\begin{theorem}[\Mech{joinRight/is-bounded}]\label{thm:cost-joinRight-is-bounded}
  Let $t_1$, $a$, and $t_2$ be valid inputs to \codename{JoinRight}.
  Then, the cost of $\Call{JoinRight}{t_1, a, t_2}$ is bounded by $1 + 2(n_1 - n_2)$.
\end{theorem}

\begin{proof}
  We prove a strengthened claim:
  \begin{enumerate}
    \item \label{lem:cost-joinRight-is-bounded-red} If $t_1$ is colored $\colorred$, the cost of $\Call{JoinRight}{t_1, a, t_2}$ is bounded by $1 + 2(n_1 - n_2)$.
    \item \label{lem:cost-joinRight-is-bounded-black} If $t_1$ is colored $\colorblack$, the cost of $\Call{JoinRight}{t_1, a, t_2}$ is bounded by $2(n_1 - n_2)$.
  \end{enumerate}
  The desired result follows immediately in both cases.
  Following the structure of the \codename{JoinRight} in \cref{alg:join-right}, we go by induction on $t_1$.

  \begin{enumerate}[label=\Roman*.]
    \item
      Since $t_1$ is colored $\colorred$, $t_{1,2}$ is black with $n_1 = n_{1,2}$, and we must prove \cref{lem:cost-joinRight-is-bounded-red}.
      This case incurs \costbox{$1$} cost in addition to the cost of the recursive call.
      The cost of the recursive call is bounded by $2(n_{1,2} - n_2) = 2(n_1 - n_2)$.
      Therefore, the cost of the entire computation is bounded by $1 + 2(n_1 - n_2)$, as desired.

    \item This case incurs zero cost.
    \item This case incurs zero cost.
    \item This case incurs zero cost.

    \item
      Since $t_1$ is colored $\colorblack$, $n_1 = n_{1,2} + 1$, and we must prove \cref{lem:cost-joinRight-is-bounded-black}.
      This case incurs \costbox{$1$} cost in addition to the cost of the recursive call.
      The color of $t_{1,2}$ is unknown, but in either case the cost of the recursive call is bounded by $1 + 2(n_{1,2} - n_2)$.
      Therefore, the cost of the entire computation is bounded by $2 + 2(n_{1,2} - n_2) = 2((n_{1,2} + 1) - n_2) = 2(n_1 - n_2)$, as desired.

    \item
      This case is the same as the previous case.
  \end{enumerate}
  In all cases, the desired result holds.
\end{proof}

\subsection{Cost of \seqjoin}

Using \cref{thm:cost-joinRight-is-bounded}, we may now reason about the cost of the full $\seqjoin$ implementation of \cref{alg:join}.
For notational convenience, we write
\begin{align*}
  \overline{x_1} &= \max(x_1, x_2) \\
  \overline{x_2} &= \min(x_1, x_2)
\end{align*}
since $\seqjoin$ behaves symmetrically depending on which tree is larger.

\begin{theorem}[\Mech{join/is-bounded}]\label{thm:cost-join-is-bounded}
  For all $t_1$, $a$, and $t_2$,
  the cost of $\Call{Join}{t_1, a, t_2}$ is bounded by $1 + 2(\overline{n_1} - \overline{n_2})$.
\end{theorem}
\begin{proof}
  If $t_1$ and $t_2$ have the same black height, no cost is incurred, so the bound is trivially met.
  Otherwise, the result follows immediately from \cref{thm:cost-joinRight-is-bounded}.
\end{proof}

This validates the claim by \citet[\S 4.2]{Blelloch2022} that the cost of \seqjoin{} on red-black tree is in $\bigO{\length{h(t_1) - h(t_2)}}$, where $h(t)$ is the height of tree $t$.

Since black height is a property only understood in the implementation, rather than the abstract sequence interface, we wish to publicly characterize the cost of \seqjoin{} in terms of the lengths of the involved sequences.
To accomplish this, we bound the black height of a red-black tree in terms of the overall size of the tree, which we write $\length{t}$ for a tree $t$.

\begin{lemma}[\Mech{nodes/upper-bound}]\label{lem:nodes-upper-bound}
  For any red-black tree $t$ with black height $n$, we have
  $$n \le \ceil*{\log_2 (1 + \length{t})}.$$
\end{lemma}

\begin{lemma}[\Mech{nodes/lower-bound}]
  For any red-black tree $t$ with black height $n$, we have
  $$\floor*{\frac{\ceil*{\log_2(1 + \length{t})} - 1}{2}} \le n.$$
\end{lemma}

Using these lemmas, we may give a user-facing description of the cost of \seqjoin{}.

\begin{theorem}[\Mech{join/is-bounded/nodes}]
  Let $t_1$, $a$, and $t_2$ be valid inputs to \codename{Join}.
  Then, the cost of $\Call{Join}{t_1, a, t_2}$ is bounded by
  $$1 + 2\left(\ceil*{\log_2(1 + \overline{\length{t_1}})} - \floor*{\frac{\ceil*{\log_2(1 + \overline{\length{t_2}})} - 1}{2}}\right).$$
\end{theorem}

This matches the expected cost bound, $$\bigO{\ceil*{\log_2\left(\overline{\length{t_1}} / \overline{\length{t_2}}\right)}}.$$

\section{Case study: algorithms on sequences}

An essential part of the work of \citet{Blelloch2016,Blelloch2022} and \citet{Sun2019} is showing how an implementation of the sequence signature gives rise to efficient implementations of other common algorithms on sequences when sequences are implemented as balanced trees.
Here, we consider the implementation and cost analysis of some such algorithms.
We implement each algorithm generically in terms of the sequence interface given in \cref{fig:seq-sig-cbpv}.
However, for the purpose of cost analysis, we break abstraction, inlining the sequence definitions.
Additionally, for readability, we replace uses of $\seqrec$ with a more familiar pattern matching notation.

\subsection{Sequence sum}

One simple algorithm on a sequence of natural numbers is a parallel sum, adding up the elements in linear work and logarithmic span with respect to the length of the sequence when counting recursive calls.
We give an implementation
$$\codename{Sum} : \seqty[\nat] \to \F{\nat}$$
in \cref{alg:sum}, adapting the definition from \cref{code:seq-examples} to the call-by-push-value setting and adding cost instrumentation and parallelism.
It goes by recursion using $\seqrec[\F{\nat}]$.
In the base case, $0$ is returned.
In the inductive case, it recursively sums both subsequences \emph{in parallel} and then returns the sum of the results and the middle datum.

\begin{algorithm}
  \begin{algorithmic}
  \Require
    \Statex $s : \seqty[\nat]$
  \Ensure
    \Statex $\Call{Sum}{s} : \F{\nat}$
  \Statex
    \Switch{$s$}
    \Case{$\seqempty$}
      \State \Return $0$
    \EndCase
    \Case{$\seqjoin{}(s_1, a, s_2)$}
      \Step 1
      \State $(x_1, x_2) \gets \para{\Call{Sum}{s_1}}{\Call{Sum}{s_2}}$
      \State \Return $x_1 + a + x_2$
    \EndCase
  \end{algorithmic}
  \caption{%
    Recursive \codename{Sum} algorithm for sequences.
    Pattern-matching syntax for $\seqempty$ and $\seqjoin$ is syntactic sugar for $\seqrec[\F{\nat}]$.
  }
  \label{alg:sum}
\end{algorithm}

When the implementation of sequences is specialized to red-black trees, we achieve the desired cost bound.

\begin{theorem}[\Mech{sum/bounded}]
  For all red-black trees $t$, the cost of $\Call{Sum}{t}$ is bounded by
  \begin{itemize}
    \item $\length{t}$ work (sequential cost) and
    \item $1 + 2\ceil*{\log_2 (1 + \length{t})}$ span (idealized parallel cost).
  \end{itemize}
\end{theorem}
\begin{proof}
  The sequential bound is immediate by induction.
  The parallel bound is shown using the black height, showing a bound of $1 + 2n$ (and a strengthened bound of $2n$ in case the tree is black) by induction.
  Then, \cref{lem:nodes-upper-bound} translates the bound from black height to the size of the tree.
\end{proof}

This matches the result of \citet{Blelloch2016,Blelloch2022}: linear work and logarithmic span.

\subsection{Finite set functions}\label{sec:finite-set}

\citet{Blelloch2016,Blelloch2022} consider implementations of standard functions on finite sets using balanced trees.
Here, we briefly show how such implementations could be provided in terms of the basic sequence signature of \cref{fig:seq-sig-cbpv}.

In order to implement a finite set as a sequence, we assume the element type $\alpha$ is equipped with a total order.
Then, standard functions on finite sets may be implemented using the recursor on sequences.
In \cref{fig:seq-more-examples}, we provide generic implementations of some examples:
\begin{enumerate}
  \item
    The \codename{Split} function splits a sorted sequence at a designated value, providing the elements of the sequence less than and greater than the value and, if it exists, the equivalent value.
  \item
    The \codename{Insert} function inserts a new value into the correct position in a sorted sequence, simply splitting the sequence at the desired value and joining the two sides around the new value.
  \item
    The \codename{Union} function takes the union of two sorted sequences, combining their elements to make a new sorted sequence.
\end{enumerate}
\citeauthor{Blelloch2016} study the efficiency of these and other similar algorithms is studied, showing that implementations in terms of $\seqempty$, $\seqjoin$, and $\seqrec$ have comparable efficiency to bespoke definitions.
We include the implementations of these algorithms in our mechanization, but we leave their cost and correctness verification to future work.

\begin{figure}
  \iblock{
    \mrow{\codename{Split} : \seqty[\alpha] \to \alpha \to \F{\prodty{\seqty[\alpha]}{\prodty{\optionty{\alpha}}{\seqty[\alpha]}}}}
    \mrow{\codename{Split}~s~a = }
    \mrow{\quad \seqrec[\F{\prodty{\seqty[\alpha]}{\prodty{\optionty{\alpha}}{\seqty[\alpha]}}}]}
    \mrow{\quad \quad \ret{\seqempty, \noneexp, \seqempty}}
    \mrow{\quad \quad (\lambda~s_1~r_1~a'~s_2~r_2.}
    \mrow{\quad \quad \quad \compare{a}{a'} }
    \mrow{\quad \quad \quad \quad \equal: \ret{s_1, \someexp{a}, s_2}}
    \mrow{\quad \quad \quad \quad \less: (s_{1,1}, a?, s_{1,2}) \leftarrow r_1;}
    \mrow{\quad \quad \quad \quad \quad\ \ s' \leftarrow \seqjoin(s_{1,2}, a', s_2);}
    \mrow{\quad \quad \quad \quad \quad\ \ \ret{s_{1,1}, a?, s'}}
    \mrow{\quad \quad \quad \quad \greater: (s_{2,1}, a?, s_{2,2}) \leftarrow r_2;}
    \mrow{\quad \quad \quad \quad \quad\ \ s' \leftarrow \seqjoin(s_1, a', s_{2,1});}
    \mrow{\quad \quad \quad \quad \quad \ \ \ret{s', a?, s_{2,2}})}
    \mrow{\quad \quad s}
  }
  \iblock{
    \mrow{\codename{Insert} : \seqty[\alpha] \to \alpha \to \F{\seqty[\alpha]}}
    \mrow{\codename{Insert} ~s~a = (s_1, a?, s_2) \leftarrow \codename{Split}~s~a; ~\seqjoin(s_1, a, s_2)}
  }
  \iblock{
    \mrow{\codename{Union} : \seqty \to \seqty \to \F{\seqty}}
    \mrow{\codename{Union} = }
    \mrow{\quad \seqrec[\seqty \to \F{\seqty}]}
    \mrow{\quad \quad (\lambda s.\ s)}
    \mrow{\quad \quad (\lambda~\textunderscore~f_1~a~\textunderscore~f_2.\ \lambda s_2.}
    \mrow{\quad \quad \quad (s_{2,1}, a?, s_{2,2}) \leftarrow \codename{Split}~s_2~a;}
    \mrow{\quad \quad \quad (u_1, u_2) \leftarrow \para{f_1~s_{2,1}}{f_2~s_{2,2}};}
    \mrow{\quad \quad \quad \seqjoin(u_1, a', u_2))}
  }

  \caption{Sample implementations of functions on sequences that use $\seqempty$, $\seqjoin$, and $\seqrec[\rho]$.}
  \label{fig:seq-more-examples}
\end{figure}
\section{Conclusion}\label{sec:concl}

In the work, we presented an implementation of the $\seqjoin$ algorithm on red-black trees \citet{Blelloch2016,Blelloch2022} whose correctness is intrinsically verified due to structural invariants within the type definition.
Our implementation was given in \Calf{}, instrumented with cost annotations to count the number of recursive calls performed; using the techniques developed by \citet{Niu2022}, we gave a formally verified precise cost bound proof for the $\seqjoin$ algorithm.

As noted by \citet{Blelloch2016,Blelloch2022}, balanced trees are an appealing choice for the implementation of persistent sequences.
Since the $\seqjoin$-based presentation of sequences provides an induction principle over the underlying balanced trees, where call-by-push-value suspends the results of recursive calls, we were able to implement standard functional algorithms on sequences and, following \citeauthor{Blelloch2016}, prove their efficient sequential and parallel cost bounds.

\subsection{Future work}

In this work, we begin to study parallel-ready data structures.
This suggests a myriad of directions for future work.

\paragraph{Full sequence library}
A natural next step following from this work would be the verification of correctness conditions and cost bounds on other algorithms included in persistent sequence libraries.

\paragraph{Finite sets and dictionaries}
Another common use case of balanced trees, as explored in depth by \citet{Blelloch2016,Blelloch2022}, is the implementation of finite sets and dictionaries by imposing and maintaining a total order on the data stored in the tree.
In \cref{sec:finite-set}, we briefly discuss the implementation of finite sets using sorted sequences; as future work, we hope to extend this development to a full-scale finite set library with cost and correctness verification.

\paragraph{Amortized complexity}
Although we study the binary \seqjoin{} operation on red-black trees in this work, more common historically is the single-element insertion operation.
Once the desired location for the new element is found, insertion into the tree along with any necessary rebalancing has asymptotically constant amortized cost \cite{tarjan:1983}.
We expect this result could be verified similarly to other amortized analyses in \Calf{} \cite{grodin-harper:2023}.

\paragraph{Various balancing schemes}
\citet{Blelloch2016,Blelloch2022} study a variety of tree balancing schemes, including AVL trees, weight-balanced trees, and treaps.
All of these balancing schemes match the sequence signature, as well; we hope to implement and verify these schemes in future work. Unlike red-black trees, some of these schemes cannot be implemented purely functionally, \eg{} treaps. This suggests an extension of \Calf{} that can better take effects into account.

\paragraph{Modular analysis of large-scale algorithms}
Many functional algorithms are implemented based on sequences, finite sets, and dictionaries \cite{ab-algorithms}.
However, in this work, we were forced to reveal the implementation of sequences as red-black trees in order to analyze the efficiency of algorithms implemented generically, such as \codename{Sum}.
In general, such analyses may even depend on particular hidden invariants within an implementation type; thus, we anticipate that analysis of larger-scale algorithms in this fashion would be intractable.
Going forward, we hope to further develop a theory of modularity for algorithm cost, allowing algorithms implemented in terms of abstract data types to be analyzed without fully revealing the implementation of the abstraction.
 
\begin{acks}
  We are grateful to Guy Blelloch for insightful discussions and advice about this work.

  This work was supported in part by the \grantsponsor{AFOSR}{United States Air Force Office of Scientific Research}{https://www.afrl.af.mil/AFOSR/} under grant number \grantnum{AFOSR}{FA9550-21-0009} (Tristan Nguyen, program manager) and the \grantsponsor{NSF}{National Science Foundation}{https://nsf.gov} under award number \grantnum{NSF}{CCF-1901381}.
  Any opinions, findings and conclusions or recommendations expressed in this material are those of the authors and do not necessarily reflect the views of the AFOSR or the NSF.
\end{acks}

\bibliography{my}

\end{document}